\documentclass[journal]{IEEEtran}

\ifCLASSINFOpdf
 \usepackage[pdftex]{graphicx}
\else
 \usepackage[dvips]{graphicx}
\fi

\usepackage[cmex10]{amsmath}
\usepackage{cite}
\usepackage{balance}
\usepackage{algorithmic}
\usepackage{algorithm}
\usepackage{mdwlist}
\usepackage{array}
\usepackage{multirow}
\usepackage{mdwmath}
\usepackage{mdwtab}
\usepackage{soul}
\usepackage{epstopdf}
\usepackage{amsthm}
\newtheorem{thm}{Theorem}
\newtheorem{lem}{Lemma}
\newtheorem{defn}{Definition}

\usepackage{varioref}
\usepackage[normalem]{ulem}
\usepackage{amssymb}
\usepackage[pagewise]{lineno}

\begin{document}

\title{Radio Resource Allocation for Multicarrier-Low Density Spreading Multiple Access}

\author{Mohammed~Al-Imari,~\IEEEmembership{Member,~IEEE,} Muhammad~Ali~Imran,~\IEEEmembership{Senior~Member,~IEEE,} and~Pei~Xiao,~\IEEEmembership{Senior~Member,~IEEE}%
\thanks{M. Al-Imari is with Samsung R\&D Institute UK, Staines-upon-Thames, Surrey, TW18 4QE, UK (e-mail: m.al-imari@samsung.com).}%
\thanks{M. Imran and P. Xiao are with the Institute for Communication Systems, home of the 5G Innovation Centre, University of Surrey, Guildford, Surrey, GU2 7XH, UK (e-mail: \{m.imran, p.xiao\}@surrey.ac.uk).}%
}

\markboth{}%
{Shell \MakeLowercase{\textit{et al.}}: Bare Demo of IEEEtran.cls for Journals}

\maketitle

\begin{abstract}
Multicarrier-low density spreading multiple access (MC-LDSMA) is a promising multiple access technique that enables near optimum multiuser detection. In MC-LDSMA, each user's symbol spread on a small set of subcarriers, and each subcarrier is shared by multiple users. The unique structure of MC-LDSMA makes the radio resource allocation more challenging comparing to some well-known multiple access techniques. In this paper, we study the radio resource allocation for single-cell MC-LDSMA system. Firstly, we consider the single-user case, and derive the optimal power allocation and subcarriers partitioning schemes. Then, by capitalizing on the optimal power allocation of the Gaussian multiple access channel, we provide an optimal solution for MC-LDSMA that maximizes the users' weighted sum-rate under relaxed constraints. Due to the prohibitive complexity of the optimal solution, suboptimal algorithms are proposed based on the guidelines inferred by the optimal solution. The performance of the proposed algorithms and the effect of subcarrier loading and spreading are evaluated through Monte Carlo simulations. Numerical results show that the proposed algorithms significantly outperform conventional static resource allocation, and MC-LDSMA can improve the system performance in terms of spectral efficiency and fairness in comparison with OFDMA.
\end{abstract}
\begin{IEEEkeywords}
Low density spreading, uplink communications, radio resource allocation, fairness.
\end{IEEEkeywords}

\IEEEpeerreviewmaketitle
\section{Introduction}\label{SectionOne}
The need for ubiquitous coverage and the increasing demand for high data rate services keep constant pressure on the mobile network infrastructure. There has been intense research to improve the system spectral efficiency and coverage, and a significant part of this effort focused on developing and optimizing the multiple access techniques. One such promising technique that been recently proposed is the low density spreading (LDS), which intelligently manages the multiple access interference to offer efficient and low complexity multiuser detection (MUD)~\cite{Tse_LDS}. It was shown that the low density spreading multiple access (LDSMA) enables near optimal MUD using belief propagation with affordable complexity~\cite{Tse_LDS,YoshidaTanaka_SSCDMA,RaymondSaad_SSCDMA,GuoWang_SSCDMA}. The LDSMA concept has been extended to multicarrier communication to cope with the multipath channel effect~\cite{FerThesis}. In multicarrier-LDSMA (MC-LDSMA), each symbol is spread over a small number of the available subcarriers, and each subcarrier is shared by more than one user. The main focus in the literature of LDS, was on the evaluation of the asymptotic behaviour of the belief propagation algorithm and optimization of the system's bit error rate (BER) performance~\cite{Tse_LDS,YoshidaTanaka_SSCDMA,RaymondSaad_SSCDMA,GuoWang_SSCDMA}.

Given the limited spectrum availability, radio resource allocation is essential for optimizing the performance of mobile communication systems. Most of existing work on radio resource allocation for uplink scenario can be divided into two categories. The first category deals with resource allocation for generic multiple access channel (MAC), where only the users' power constraints are considered~\cite{Verdu_ISI,Tse_Polymatroid,WeiYu_Jour}. In the second category, specific multiple access schemes such as orthogonal frequency division multiple access (OFDMA) and time division multiple access are considered. Comparing to the generic MAC, in the specific multiple access schemes, extra constraints need to be satisfied in the resource allocation problem. In general, these constraints make the allocation problem non-linear and non-convex, which motivates suboptimal solutions, and can be resolved by analysing the optimality conditions~\cite{NgSung_SPA_UL,HuangBerry_SPA_UL,EURASIP_OFDMA} or based on the game theory approach~\cite{YPW_Auction,HJL_Nash_J,Noh_Game}. Although, there has been a lot of work on the radio resource allocation, the existing schemes cannot be directly applied to the MC-LDSMA system under question.

The LDS structure imposes constraints on the number of users that share the same subcarrier and the number of subcarriers used for spreading each symbol, which make the radio resource allocation problem more challenging. We have proposed a heuristic algorithm for radio resource allocation without considering the spreading to simplify the problem~\cite{MyPaper_VTC11}. In this paper, we consider the radio resource allocation for MC-LDSMA by taking into account the spreading, which significantly change the structure of the problem. We will consider the radio resource allocation for two cases; single-user and multiuser. The contributions of this paper can be summarized as follows. i)~A single-user optimal power allocation scheme that consists of maximum ratio transmission (MRT) and water-filling is proposed. ii)~An optimal subcarrier partitioning scheme for the proposed power allocation is derived. iii)~Assuming continuous frequency selective channels, an optimal power allocation for MC-LDSMA is derived under certain relaxed conditions. iv)~Using the single-user power allocation algorithm and the insights gained from the multiuser optimal solution, two suboptimal subcarrier and power allocation algorithms are proposed. v)~The effects of subcarrier loading, effective spreading factor and power allocation are investigated through Monte Carlo simulations.\\
The rest of the paper is organized as follows. In Section~\ref{SectionTwo}, background information on the optimal power allocation for frequency selective MAC and the MC-LDSMA system model are presented. The optimal power allocation and subcarriers partitioning schemes for the single-user case are introduced in Section~\ref{SectionThree}. In Section~\ref{SectionFour}, the optimality conditions of radio resource allocation for the multiuser case in addition to suboptimal algorithms are provided. In Section~\ref{SectionFive}, we evaluate the performance of the proposed algorithms in single-user and multiuser scenarios using Monte Carlo simulations. Finally, concluding remarks are drawn in~Section~\ref{SectionSix}. Table~\ref{table_ListSyms} summarizes the frequently used mathematical symbols in this paper.
\section{Background and System Model}\label{SectionTwo}
\begin{table}[!t]
\caption{List of Symbols.}
\centering 
\begin{tabular}{c||l} 
\hline
Symbol & Description\\
\hline
$\mathcal{K}$ & Set of users\\ 

$K$ & Total number of users\\ 

$\mathcal{N}$ & Set of subcarriers\\ 

$N$ & Total number of subcarriers\\

$\mathcal{N}_k$ & Set of subcarriers allocated to user $k$\\

$N_k$ & Number of subcarriers allocated to user $k$\\ 

$\mathcal{M}_k$ & Set of subcarriers' subsets/groups of user $k$\\

$M_k$ & Number of symbols/subcarriers' subsets of user $k$\\ 

$P_k$ & Maximum power of user $k$\\

$w_k$ & Weight assigned to user $k$\\ 

$R_k$ & Rate of user $k$\\

$\lambda_k$ & Lagrange multiplier of the $k$th user's power constraint\\ 

$\mathcal{I}_k$ & Set of users that interfere with the $k$th user\\

$q$ & Interference from users not cancelled yet\\

$\sigma^2$ & Noise power spectral density\\

$\pi(.)$ & Permutation\\ 

$d_c$ & Number of users at each subcarrier in LDS\\

$d_v$ & Effective spreading factor in LDS\\

$\mathcal{D}_{k,m}$ & Set of subcarriers to spread symbol $m$ of user $k$ in LDS\\

$R_{k,m}$ & $k$th user rate on the $m$th subset of subcarriers in LDS\\

$u_{k,m}$ & $k$th user utility on the $m$th subset of subcarriers in LDS\\

$u_{k,n}$ & $k$th user utility on the $n$th subcarrier\\

$u_k(q,f)$ & $k$th user utility on the frequency $f$ with interference $q$\\

$\mathcal{B}$ & Set of users allocated the same weight/priority\\

$\bar{w}_l$ & Weight assigned to the $l$th group of users $(\mathcal{B}_l)$\\

$\mathcal{S}_f$ \& $\mathcal{S}_n$ & Set of active users at frequency $f$ \& subcarrier $n$\\ 

$\mathcal{P}_k(f)$ \& $p_{k,n}$ & $k$th user power at frequency $f$ \& subcarrier $n$\\

$x_k(f)$/$x_{k,n}$ & Frequency/subcarrier allocation index for user $k$\\ 

$\alpha_k(f)$ \& $\alpha_{k,n}$ & $k$th user channel response at frequency $f$ \& subcarrier $n$\\

$h_k(f)$ \& $h_{k,n}$ & $k$th user channel gain at frequency $f$ \& subcarrier $n$\\
\hline
\end{tabular}
\label{table_ListSyms} 
\end{table}
The uplink scenario can be represented by a multiple access channel, where a set of users $\mathcal{K}=\{1, \dots, K\}$ transmit to a single base-station in the presence of additive Gaussian noise over frequency selective channels. As the channel frequency responses are not flat, it is expected that dynamic power allocation will improve the system performance. Considering perfect channel knowledge at the transmitter, the goal is to allocate the optimal transmit power for each user that maximizes the weighted sum-rate. The weights are used to prioritize the users as a fairness mechanism. The users are subject to individual maximum power constraints, $\mathbf{P}=[P_1, \dots, P_K]$. The capacity region of frequency selective MAC is given by\vspace{-0.1cm}
\begin{multline}\label{Eq_MAC_CapArea}
\hspace{-0.5cm}\bigcup_{\mathcal{P}\in\mathcal{F}}\Bigg\{\textbf{R}:\hspace{-0.1cm}\sum_{k\in \mathcal{S}} R_k\leq\hspace{-0.15cm}\int_{-W/2}^{W/2}\log \left(1+\frac{1}{\sigma^2}\displaystyle\sum_{k\in \mathcal{S}} \mathcal{P}_k(f) h_k(f)\right)df,\\ 
\forall \mathcal{S}\subset \mathcal{K}\Bigg\},
\end{multline}
where $\mathcal{S}$ is any subset of the total set of the users $\mathcal{K}$, and $\mathcal{F}$ is the set of feasible power allocation policies that satisfy the power constraint
\begin{equation}
\mathcal{F}\equiv \left\{\mathcal{P}:\int_{-W/2}^{W/2} \mathcal{P}_k(f)df\leq P_k, \forall k\right\}.
\end{equation}
Here, $W$ is the total bandwidth and $\mathcal{P}$ is a power allocation policy such that $\mathcal{P}_k(f)$ is the power user $k$ allocates at frequency $f$, and $h_k(f)=|\alpha_k(f)|^2$, where $\alpha_k(f)$ is the $k$th user channel frequency response at frequency $f$ and $\sigma^2$ is the noise power spectral density. This capacity area has $K!$ vertices in the positive quadrant, and each vertex is achievable by a successive decoding using one of the possible $K!$ distinct permutations of the set $\mathcal{K}$. The vertex $\mathbf{R}_\pi$, which corresponds to permutation $\pi(.)$, is given by
\begin{equation}
R_{\pi(k)}=\int_{-W/2}^{W/2} \log\Bigg( 1+\frac{\mathcal{P}_{\pi(k)}(f) h_{\pi(k)}(f)}{\sigma^2+\displaystyle\sum_{i=k+1}^{K}\mathcal{P}_{\pi(i)}(f) h_{\pi(i)}(f)}\Bigg)df.
\end{equation}
The problem of finding the optimum power allocation can be formulated as follows.\\
\textbf{Problem MAC}: \textit{Weighted sum-rate maximization for frequency selective MAC}
\begin{equation}\label{Eq_MACfs_Obj} 
\max_{\mathcal{P}}\quad \sum_{k\in \mathcal{K}} w_{\pi(k)}R_{\pi(k)},
\end{equation}
subject to
\begin{equation}\label{Eq_MACfs_PowerConst}
\int_{-W/2}^{W/2} \mathcal{P}_k(f)df\leq P_k,\quad \forall k\in \mathcal{K},
\end{equation}
where $w_k$ is the weight associated with user $k$. By exploiting the Polymatroid properties of the capacity region~\eqref{Eq_MAC_CapArea} and using the Lagrange multipliers, Tse and Hanly~\cite{Tse_Polymatroid} have shown that optimal power allocation can be carried out by optimization at each frequency, and the optimal solution for a given $\boldsymbol{\lambda}=[\lambda_1, \dots, \lambda_K]$ is given by
\begin{multline}
\sum_{k\in\mathcal{K}} w_{\pi(k)} \log\left( 1+\frac{\mathcal{P}^\star_{\pi(k)}(f) h_{\pi(k)}(f)}{\sigma^2+\displaystyle\sum_{i=k+1}^{K}\mathcal{P}^\star_{\pi(i)}(f) h_{\pi(i)}(f)}\right)\\-\sum_{k\in\mathcal{K}} \lambda_k \mathcal{P}^\star_k(f)= \int_0^\infty u^\star(q,f)dq,
\end{multline}
where
\begin{equation}\label{Eq_UtilityMax}
u^\star(q,f)=\Big[\max_{k\in\mathcal{K}}u_k(q,f)\Big]^+,
\end{equation}
\begin{equation}
u_k(q,f)=\frac{w_k}{\sigma^2+q}-\frac{\lambda_k}{ h_k(f)}.
\end{equation}
Here, $u_k(q,f)$ is the utility of user $k$ on frequency $f$, and $q$ is the residual interference due to the users that are not yet cancelled. $\lambda_k$ is the Lagrange multiplier associated with the power constraint of user $k$. The optimal power allocation can be obtained via a greedy algorithm, which is summarized as follows~\cite{MyPaper_VTC11}. Let $\sigma^2+q$ represents the current ``interference level'' due to the noise and received power of users not yet cancelled. Starting with received power $q=0$, the optimal solution can be obtained by adding a user, at each interference level, that leads to the largest increase in the objective function. The decoding order that achieves the optimal solution is in increasing order of the users' weights $w_{\pi(1)}\leq w_{\pi(2)}\dots \leq w_{\pi(K)}$. It is worth mentioning that the solution is only valid for continuous channel gain or when the number of discrete channel gains reaches infinity~\cite{Tse_Polymatroid}.

For MAC, users' total transmit power $(P_k)$ is the only constraint in the optimization problem. For special multiple access schemes, an additional constraint is imposed on the optimization problem to control the number of users that share the same frequency. Let $\mathcal{S}_f$ be the set of active users at frequency $f$, $\mathcal{S}_f=\{k: \mathcal{P}_k(f)>0\}$. As a special multiple access scheme, in MC-LDSMA the condition, $\vert \mathcal{S}_f\vert\leq d_c$, must be satisfied. On the other hand, in OFDMA $\vert \mathcal{S}_f\vert\in\{0, 1\}$ must be satisfied. As will be shown in Section~IV, adding such a constraint to the optimization problem makes the problem non-convex. However, it will be proven that MC-LDSMA can be a special case of the generic MAC with certain conditions on the users' weights, and the power allocation algorithm for MAC can be used to find the optimal power allocation for MC-LDSMA. Therefore, the constraint will be fulfilled without explicitly imposing it in the optimization problem.

\begin{figure}[!t] 
\centering
\includegraphics[width=8cm]{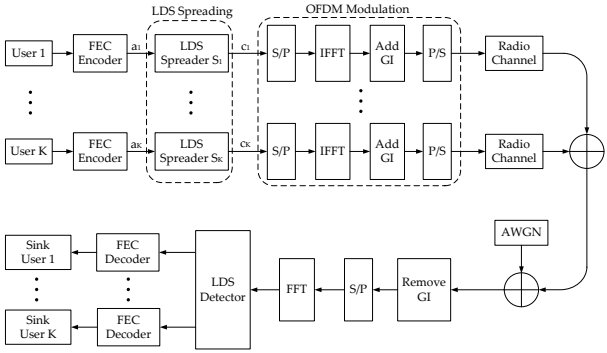}
\caption{Uplink MC-LDSMA block diagram.}
\label{Figure1_LDSbd}
\end{figure}

Now, the system model of a single-cell MC-LDSMA will be presented, which can also be found in~\cite{MyPaper_IWCMC}. The conceptual block diagram of an uplink MC-LDSMA system is depicted in Fig.~\ref{Figure1_LDSbd}. The total frequency band is divided into a set of subcarriers $\mathcal{N}=\{1, \dots, N\}$. Let $\textbf{a}_k$ be a data vector of user $k$ consisting of $M_k$ modulated data symbols and denoted as $\textbf{a}_k = [a_{k,1}, a_{k,2}, \dots, a_{k,M_k}]^T$. The signature matrix $\textbf{S}_k$ assigned for the $k$th user consists of $M_k$ LDS signatures, i.e. $\textbf{S}_k = [\textbf{s}_{k,1}, \textbf{s}_{k,2} \dots, \textbf{s}_{k,M_k}]$. Each LDS signature, $\textbf{s}_{k,m}\in\mathbb{C}^{N\times1}$, is a sparse vector consisting of $N$ chips among which only $d_v$ chips have non-zero values, where $d_v$ is the effective spreading factor. Each data symbol $a_{k,m}$ will be spread using the $m$th spreading sequence. Let $\textbf{c}_k = [c_{k,1}, c_{k,2}, \dots, c_{k,N}]^T$ denotes the chips vector belonging to user $k$ after the spreading process, which is given by $\textbf{c}_k = \textbf{S}_k \textbf{a}_k$.

At the receiver side, users' signals that are using the same subcarrier will be superimposed. However, the number of users interfere at each subcarrier $(d_c)$ is much less than the total number of users, i.e. $d_c\ll K$. After performing OFDM demodulation operation at the receiver, the received signal at subcarrier $n$ is given by
\begin{equation}\label{Eq_LDS_RecSign}
y_n = \sum_{k\in\mathcal{S}_n} \alpha_{k,n}c_{k,n} + v_n,
\end{equation}
where $v_n$ is the additive white Gaussian noise (AWGN) and $\mathcal{S}_n$ is the set of users that use subcarrier $n$. This signal is passed to LDS MUD to separate users' symbols using close to optimum chip-level iterative MUD based on message passing algorithm. More details about the LDS receiver can be found in~\cite{Tse_LDS}.
\section{Single-User Radio Resource Allocation}\label{SectionThree}
In this section, the single-user power allocation with LDS will be considered. The power allocation for LDS systems is more challenging compared with the non-spreading case, which has a well-known water-filling solution. Assuming that a set of subcarriers, $\mathcal{N}_k=\{1, \dots, N_k\}$, are allocated to a user, the user will partition this set of subcarriers into $M_k$ subsets/groups indexed by $\mathcal{M}_k=\{1, \dots, M_k\}$. Each subset of subcarriers will be used to spread one symbol, and the total number of transmitted symbols will be $M_k=N_k/d_v$. Let the $m$th subset of subcarriers denoted by $\mathcal{D}_{k,m}$, where $|\mathcal{D}_{k,m}|=d_v$, is used to spread the $m$th symbol, then the rate achieved on that subset is given by
\begin{equation}\label{Eq_SymRate}
R_{k,m}=\log{\left(1+\left(\sum_{n\in\mathcal{D}_{k,m}} \sqrt{p_{k,n}g_{k,n}}\right)^2\right)},\quad\forall m\in\mathcal{M}_k,
\end{equation}
where $g_{k,n}=\frac{h_{k,n}}{\sigma^2W/N}$. Note that~\eqref{Eq_LDS_RecSign} represents the received signal from all the users in each subcarrier before the LDS detector (belief propagation algorithm). It is worth mentioning that in multiple access techniques that don't use spreading, such as OFDMA, the information/symbols on each subcarrier are independent, hence, the rate achieved on a subset of subcarriers is obtained by a sum of the rate achieved on each subcarrier. However, this is not the case for spreading techniques, like MC-LDSMA. With spreading the same information (modulation symbol) is repeated on a set of subcarriers. Thus, the rate achieved on a set of subcarriers will be a function of the effective signal to noise ratio (SNR) on the subcarrier. Furthermore, in the optimization problem, we consider the effective SNR on the subcarriers that are used by each user/symbol after the LDS detector, which is represented in~\eqref{Eq_SymRate}.

The single-user power allocation problem for LDS can be split into two parts: Firstly, what is the optimal subcarriers partitioning that maximizes the user rate? Secondly, for a given subcarriers partitioning, what is the optimal power allocation? These two issues will be addressed in the sequel.
\subsection{Power Allocation}
For a given subcarriers partitioning scheme, the power allocation problem can be formulated as
\begin{equation}\label{Eq_SU_MaxObj}
\max_{p_{k,n}} \sum_{m\in\mathcal{M}_k}\log{\left(1+\left(\sum_{n\in\mathcal{D}_{k,m}} \sqrt{p_{k,n} g_{k,n}}\right)^2\right)},
\end{equation}
subject to
\begin{equation}
\sum_{m\in\mathcal{M}_k}\sum_{n\in\mathcal{D}_{k,m}} p_{k,n} = P_k,\quad \text{and}\quad p_{k,n}\geq 0.
\end{equation} 
To find the optimal power allocation, first the power allocation within each symbol is considered then the power allocation over the symbols is optimized. Let $\bar{p}_{k,m}=\sum_{n\in\mathcal{D}_{k,m}}p_{k,n}$ be the power allocated to symbol $m$. The optimal power allocation of the symbol power, $\bar{p}_{k,m}$, over the symbol's subcarriers, $\mathcal{D}_{k,m}$, is given by maximum ratio transmission (MRT)
\begin{equation}\label{Eq_MRTPowerAllocation}
p_{k,n}=\frac{g_{k,n}}{\sum_{n\in\mathcal{D}_{k,m}} g_{k,n}}\ \bar{p}_{k,m},\quad \forall n\in\mathcal{D}_{k,m}.
\end{equation}
The MRT is optimum in the sense it maximizes the received SNR of the symbol~\cite[Chapter 7]{WC_GS_Book}. Consequently, by substituting~\eqref{Eq_MRTPowerAllocation} in the objective function~\eqref{Eq_SU_MaxObj}, the optimization problem will be
\begin{equation}\label{Eq_SymbolsRate}
\max_{\bar{p}_{k,m}} \sum_{m\in\mathcal{M}_k}\log{(1+ \bar{p}_{k,m}\bar{g}_{k,m})},
\end{equation} 
subject to
\begin{equation}\label{Eq_SymPowerConst}
\sum_{m\in\mathcal{M}_k} \bar{p}_{k,m} = P_k,\quad \text{and}\quad \bar{p}_{k,m}\geq 0,
\end{equation}
where
\begin{equation}\label{Eq_EffecSymGain}
\bar{g}_{k,m}=\displaystyle\sum_{n\in\mathcal{D}_{k,m}}g_{k,n},
\end{equation}
which can be solved by single-user water-filling algorithm. The water-filling solution can be expressed as follows
\begin{equation}\label{Eq_SymbolWaterFilling}
\bar{p}^\star_{k,m}=\left[\frac{1}{\lambda_k}-\frac{1}{\bar{g}_{k,m}}\right]^+,
\end{equation}
where $[x]^+=\max (0,x)$ and $\lambda_k$ is the Lagrange multiplier that should satisfy the power constraint in~\eqref{Eq_SymPowerConst}. Hence, the single-user power allocation for MC-LDSMA will be conducted in two steps: the first step is water-filling to find the power for each symbol; the second step is MRT of the symbol power over the symbol's subcarriers. This power allocation algorithm will be referred to as MRT-WF.
\subsection{Subcarriers Partitioning}
In the MC-LDSMA technique, the allocated subcarriers to a user, $\mathcal{N}_k$, should be partitioned into $M_k$ subsets, each of which will be used to transmit one symbol. In the previous section, the optimal power allocation scheme for a given subcarriers partitioning has been derived. In this section, we investigate how to optimize the partitioning of available subcarriers. In conventional LDS codes, the subcarriers are partitioned randomly. The random partitioning has shown satisfactory BER performance~\cite{GuoWang_SSCDMA,FerThesis}. However, the partitioning scheme has not yet been investigated from the user rate optimization perspective. Here, the subcarriers partitioning will be studied under the objective of user rate maximization. The problem of partitioning the subcarriers to maximize the user rate can be formulated as follows
\begin{equation}\label{Eq_SU_PartOptiProb}
\max_{\mathcal{D}_{k,m}\in\mathcal{N}_k} \sum_{m\in\mathcal{M}_k}R_{k,m}(\mathcal{D}_{k,m}),
\end{equation}
subject to
\begin{equation}
\mathcal{D}_{k,m} \cap\mathcal{D}_{k,j}=\emptyset,\quad \forall m\neq j,\ m, j\in\mathcal{M}_k,
\end{equation}
\begin{equation}\label{Eq_SU_CardConst}
|\mathcal{D}_{k,m}|=d_v, \quad\forall m\in\mathcal{M}_k.
\end{equation}
This is a combinatorial optimization problem with a large search space. The number of possible LDS codes to be generated from the $N_k$ subcarriers for a specific spreading factor $d_v$ is given by
\begin{equation}
C_{LDS}=\frac{1}{2}{2d_v\choose d_v}\sum_{i=0}^{M_k-3} {N_k-1-i d_v\choose d_v-1}.
\end{equation}
In fact, $C_{LDS}$ represents the cardinality of the feasible search space of problem~(\ref{Eq_SU_PartOptiProb}~-~\ref{Eq_SU_CardConst}). In order to see how large is the search space, let us consider $N_k=32$ and $d_v=2$. In this case, the number possible LDS codes will be $C_{LDS}=1.92\times 10^{17}$. Apparently, brute-force search is infeasible for practical systems, and partitioning schemes with low complexity need to be sought. Even though the subcarriers partitioning problem is non-convex, there is an underlying Schur-convex structure, which can be utilized to solve the problem.

Let $\mathbf{g}_k=[\bar{g}_{k,1}, \dots, \bar{g}_{k,M_k}]$ represents the vector of symbols' gains for a given subcarrier partitioning scheme. Using the MRT-WF allocation and by substituting the optimal power~\eqref{Eq_SymbolWaterFilling} in the user rate~\eqref{Eq_SymbolsRate}, the user rate can be formulated as a function of the symbol gains as follows
\begin{equation}\label{Eq_SymbolsRate_2}
R_k(\mathbf{g}_k)= \sum_{m\in\mathcal{M}_k}\log{\left(\frac{\bar{g}_{k,m}}{\lambda_k}\right)},
\end{equation}
where $\frac{1}{\lambda_k}$ is the water-level and it is given by
\begin{equation}\label{Eq_WaterLevel_1}
\frac{1}{\lambda_k}=\frac{1}{M_k}\left(P_k+\sum_{m\in\mathcal{M}_k}\frac{1}{\bar{g}_{k,m}}\right).
\end{equation}
Based on this new formulation of the user rate, the following theorem can be introduced to link between the majorization of the symbol gains and the user rate with MRT-WF power allocation.
\begin{thm}\label{ThrmRateSchur}
The user rate $R_k(\mathbf{g}_k)$ is Schur-convex function of the symbol gains vector $\mathbf{g}_k$, and
\begin{equation}
\mathbf{g}_k\succ\hat{\mathbf{g}}_k\quad \Rightarrow\quad R_k(\mathbf{g}_k)\geq R_k(\hat{\mathbf{g}}_k),
\end{equation}
if $P_k<M_k^3-M_k^2-M_k$, where $\mathbf{x}\succ\mathbf{y}$ means $\mathbf{x}$ majorize $\mathbf{y}$.
\end{thm}
\begin{IEEEproof}
See Appendix~\ref{App_RateScurProof}
\end{IEEEproof}
Theorem~\ref{ThrmRateSchur} shows that the user rate~$R_k(\mathbf{g}_k)$ will be increased by majorizing the symbols' gains vector $\mathbf{g}_k$.
After showing the effect of symbols gains majorization on the user rate, the subcarrier partitioning can thus be derived based on this finding.
\begin{lem}\label{Lem_SC_Partitiong}
The optimal subcarrier partitioning that maximizes the user rate with MRT-WF power allocation is given by
\begin{equation}\label{Eq_GreedyScPartition}
\hspace{-0.18cm}\mathcal{D}^\star_{k,m}=\{g_{k,\pi(\hat{m})}, g_{k,\pi(\hat{m}+1)}, \dots,\ g_{k,\pi(\hat{m}+d_v-1)}\},\ \forall m\in\mathcal{M}_k,
\end{equation}
where $\hat{m} = (m-1)d_v+1$, and $\pi$ represents a permutation of the subcarrier gains in a descending order such that $g_{k,\pi(1)}\geq g_{k,\pi(2)}\geq\dots \geq g_{k,\pi(N_k)}$.
\end{lem}
\begin{IEEEproof}
See Appendix~\ref{App_SymMMVproof}
\end{IEEEproof}
In the proposed scheme, the subcarriers are sorted in descending order then the first best $d_v$ subcarriers are combined to create one symbol, the second best $d_v$ for another symbol and so on. This partitioning scheme results in symbols' gains vector $(\mathbf{g}_k)$ that majorize any other partitioning scheme. The proposed scheme will be compared with brute-force search in Section~V. For comparison purposes, another partitioning scheme will be included, which attempts to generate the least majorized vector (LMV) by making the symbols gains close to each other as much as possible. In this scheme, the subcarriers are sorted in descending order $g_{k,\pi(1)}\geq g_{k,\pi(2)}\geq\dots \geq g_{k,\pi(N_k)}$. Then, the subcarriers that have the best channel gains are combined with the subcarriers that have the least gains to create one LDS symbol. The LMV partitioning scheme can be defined as
\begin{equation}\label{Eq_PartitionSceme1}
\mathcal{D}_{k,m}=\left\{ \begin{array}{ll}
\{g_{k,\pi(m)},\ g_{k,\pi(N_k-m+1)}\}, & \text{for}\ d_v=2,\\
\\

\{g_{k,\pi(m)},\ g_{k,\pi(N_k-2m+1)}, & \\

g_{k,\pi(N_k-2m+2)}\}, & \text{for}\ d_v=3.
\end{array} \right.
\end{equation}
\section{Multiuser Radio Resource Allocation}\label{SectionFour}
In this section, the radio resource allocation for the MC-LDSMA in the multiuser case will be considered. Firstly, the subcarrier and power allocation problem for MC-LDSMA is formulated. The optimality conditions of power allocation for MC-LDSMA under relaxed conditions will be provided. This will shed light on the derivation of suboptimal subcarrier and power allocation algorithms.
\subsection{Problem Formulation}
To formulate the subcarrier and power allocation problem for MC-LDSMA system, let $x_{k,n}$ be the channel allocation index such that $x_{k,n}=1$ if subcarrier $n$ is allocated to user $k$ and $x_{k,n}=0$ otherwise. Consequently, for a given subcarrier partitioning, the weighted sum-rate maximization problem for MC-LDSMA system can be formulated as follows.\\
\textbf{Problem LDS}: \textit{Weighted sum-rate maximization for MC-LDSMA}
\begin{equation}\label{Eq_LDS_Obj}
\max_{x_{k,n},p_{k,n}} \sum_{k\in\mathcal{K}}w_{\pi(k)}\sum_{m\in\mathcal{M}_{\pi(k)}}R_{\pi(k),m},
\end{equation}
where
\begin{equation}
R_{\pi(k),m}=\log{\left(1+\left(\sum_{n\in\mathcal{D}_{\pi(k),m}}\sqrt{snr_{\pi(k),n}}\right)^2\right)},
\end{equation}
\begin{equation}
snr_{\pi(k),n}= \frac{p_{\pi(k),n}h_{\pi(k),n}}{\sigma^2W/N+\displaystyle\sum_{j=k+1}^K x_{\pi(j),n}p_{\pi(j),n}h_{\pi(j),n}},
\end{equation}
subject to
\begin{equation}
\sum_{m\in\mathcal{M}_k}\sum_{n\in\mathcal{D}_{k,m}} p_{k,n} \leq P_k, \quad \forall k\in\mathcal{K},
\end{equation}
\begin{equation}
p_{k,n}\geq 0, \quad \forall k\in\mathcal{K},\hspace{0.2cm} n\in\mathcal{N},
\end{equation}
\begin{equation}\label{Eq_UsersPerSubc}
\sum_{k\in\mathcal{K}} x_{k,n} \leq d_c, \quad \forall n\in\mathcal{N},
\end{equation}
\begin{equation}\label{Eq_Binar_Q}
x_{k,n}\in \{0,1\}, \quad\forall k\in\mathcal{K},\hspace{0.2cm} n\in\mathcal{N}.
\end{equation}
The constraints in~\eqref{Eq_UsersPerSubc} and~\eqref{Eq_Binar_Q} provide control on the number of users per subcarrier, which are referred to as the loading constraints. Note that there is an implicit spreading constraint, i.e. $|\mathcal{D}_{k,m}|>1$. The loading and spreading constraints distinguish the MC-LDSMA system from generic MAC. The \emph{Problem LDS} is non-convex for two reasons. Firstly, the binary constraint in~\eqref{Eq_Binar_Q} is a non-convex set. Secondly, the interference term in the objective function~\eqref{Eq_LDS_Obj} makes it a non-convex function. Unlike the case in OFDMA, the constraint in \eqref{Eq_Binar_Q} cannot be relaxed to take any real value in the interval $[0, 1]$ to make the problem tractable~\cite{NgSung_SPA_UL,HuangBerry_SPA_UL,EURASIP_OFDMA}. In OFDMA, when the constraint is relaxed the users are still orthogonal to each other. If the binary constraint is relaxed in LDS, all the users may interfere at each subcarrier, which violate the main concept of LDS that only $d_c$ users are allowed to interfere with each other at the same subcarrier. Furthermore, it can result in a situation when all the users transmit on all the subcarriers, leading to a fully connected graph rather than a low density one.
\subsection{Optimality Conditions for MC-LDSMA}
In this section, the optimality conditions for radio resource allocation in MC-LDSMA under relaxed conditions will be presented. The aim is to address the problem at a more fundamental level by finding the relationship between MC-LDSMA and the optimal multiple access scheme. Then the insights gained from the optimality conditions will be used to design suboptimal algorithms for subcarrier and power allocation in the next section. The spreading constraint is relaxed to find the solution to the problem that satisfies the other conditions, specifically, the loading constraints. Here, it will be shown that \emph{Problem LDS} (\ref{Eq_LDS_Obj}~-~\ref{Eq_Binar_Q}) can be solved by solving \emph{Problem MAC} (\ref{Eq_MACfs_Obj}~and~\ref{Eq_MACfs_PowerConst}) under specific conditions. To relax the spreading constraint, let $d_v=1$, and by allowing the channel gain to be continuous, the following relaxed problem is obtained.

\noindent \textbf{Problem Relaxed LDS:}
\begin{equation}\label{Eq_ReLDS_Obj}
\max_{x_k(f),\mathcal{P}_k(f)} \sum_{k\in\mathcal{K}}w_{\pi(k)}x_{\pi(k)}(f)R_{\pi(k)},
\end{equation}
where
\begin{multline}
R_{\pi(k)}=\\ \int_{-W/2}^{W/2} \log{\Bigg(1+\frac{\mathcal{P}_{\pi(k)}(f)h_{\pi(k)}(f)}{\sigma^2+\displaystyle\sum_{i=k+1}^K x_{\pi(i)}(f)\mathcal{P}_{\pi(i)}(f)h_{\pi(i)}(f)}\Bigg)df},
\end{multline}
subject to
\begin{equation}
\int_{-W/2}^{W/2} \mathcal{P}_k(f)df\leq P_k, \quad \forall k\in\mathcal{K},
\end{equation}
\begin{equation}
\mathcal{P}_k(f)\geq 0, \quad \forall k\in\mathcal{K},
\end{equation}
\begin{equation}\label{Eq_UsersPerSubcR}
\sum_{k\in\mathcal{K}} x_k(f) \leq d_c,
\end{equation}
\begin{equation}\label{Eq_Binar_QR}
x_k(f)\in \{0,1\}, \quad\forall k\in\mathcal{K}.
\end{equation}
In this new formulation, the problem is still non-convex due to the binary constraint~\eqref{Eq_Binar_QR}. However, this relaxation will allow us to satisfy the loading constraints~\eqref{Eq_UsersPerSubcR} and~\eqref{Eq_Binar_QR} without explicitly imposing them in the optimization problem as follows.
\begin{lem}\label{Lemma1}
The solution of \emph{Problem Relaxed LDS} (\ref{Eq_ReLDS_Obj}~-~\ref{Eq_Binar_QR}) can be obtained by solving \emph{Problem MAC} (\ref{Eq_MACfs_Obj}~and~\ref{Eq_MACfs_PowerConst}) with users grouped into $d_{c}$ groups and the users in each group, $\mathcal{B}_l, l=1, \dots, d_{c}$, are assigned the same weight
\begin{equation}\label{Eq_UsersGrouping}
w_k= \bar w_l,\quad \forall k\in\mathcal{B}_l,\quad l=1, \dots, d_{c}.
\end{equation}
\end{lem}
\begin{proof}
We have to prove that at each frequency, only one user from each group may be chosen for all the interference levels. With users grouping~\eqref{Eq_UsersGrouping}, the maximization problem in~\eqref{Eq_UtilityMax} can be divided into number of sub-problems, each for one group $\mathcal{B}_l$
\begin{equation}
u^*(q,f)=\max_l\Big(\max_{k\in\mathcal{B}_l}\;u_k(q,f)\Big).
\end{equation}
In each group of users $\mathcal{B}_l$, the term $\frac{\bar w_{l}}{\sigma^2+q}$ will be the same for all the users in the group and the user with the smallest $\frac{\lambda_k}{h_k(f)}$ will be dominant for each frequency~$f$

\begin{equation}
u_k(q,f)=\frac{\bar w_l}{\sigma^2+q}-\frac{\lambda_k}{h_k(f)},\quad \forall k\in\mathcal{B}_l,
\end{equation}
\begin{multline}
u_k(q,f) > u_j(q,f),\quad\text{for}\quad q\in[0,\infty),\quad if\\ \quad\frac{\lambda_k}{h_k(f)}<\frac{\lambda_j}{h_k(f)},\quad \forall j\neq k,\quad k, j\in\mathcal{B}_l.\nonumber
\end{multline}
Hence, for each frequency, only one user from each group will be selected, which is given by
\begin{equation}
k_l^*(f)=\arg\min_{k\in\mathcal{B}_l}\;\frac{\lambda_k}{h_k(f)}.
\end{equation}
Consequently, no more than $d_c$ users will share each frequency~$f$.
\end{proof}
The objective is that, by grouping the users into groups and allocating the same weight for the users in the same group, the users sharing each frequency will not exceed the number of groups. The complexity of the optimal solution lies in the computation of the Lagrange multipliers, which are computed by an iterative algorithm to satisfy the power constraints~\cite{Tse_Polymatroid}. Consequently, the complexity for solving Problem~MAC is still too high for practical cellular systems. Furthermore, partitioning of the available bandwidth in an optimal fashion is difficult to achieve practically. However, the optimal algorithm and Lemma~$1$ shed light on the structure of the optimal solution, which will be useful in designing suboptimal algorithms.
\subsection{Suboptimal Algorithms}
As discussed in the previous section, finding the optimal solution to \emph{Problem LDS} (\ref{Eq_LDS_Obj}~-~\ref{Eq_Binar_Q}) is computationally intensive and impractical. Therefore, suboptimal subcarrier and power allocation algorithms with low complexity are presented here. We propose two suboptimal solutions that differ in the stage where the spreading is considered. In the first algorithm, the effect of spreading is taken into account from the outset, and in the following steps we consider: 
\begin{itemize}
\item Each user will partition the available subcarriers, $\mathcal{N}$, into $M$ subsets using the proposed partitioning scheme~\eqref{Eq_PartitionSceme1}.
\item Each user performs MRT-WF,~\eqref{Eq_MRTPowerAllocation}~and~\eqref{Eq_SymbolWaterFilling}, as single-user power allocation strategy.
\item The subcarriers are allocated on set basis. In other words, at each iteration of the algorithm, $d_v$ subcarriers are allocated to one user.
\end{itemize}
In the second algorithm, no spreading is assumed at the start and
\begin{itemize}
\item The subcarriers will not be partitioned at the first stage.
\item Single-user water-filling power allocation is used for each user.
\item The subcarriers are allocated individually, i.e., at each iteration of the algorithm one subcarrier is allocated to one user.
\item After the subcarrier and power allocation is finished, the subcarriers are partitioned using~\eqref{Eq_PartitionSceme1} and the spreading effect is taken into account for rate calculation.
\end{itemize}
As the first algorithm is based on MRT-WF, it is referred to it as MUMRT, where MU stands for multiuser. The second algorithm is referred to as MUWF, since it is based on water-filling power allocation.

For the MUMRT algorithm, let us define the utility of the $k$th user on the $m$th set of subcarriers, $\mathcal{D}_{k,m}$, as follows\\
\begin{multline}\label{Eq_UtilityMUMRT}
u_{k,m}=\\w_k \log{\left(1+\left(\sum_{n\in\mathcal{D}_{k,m}}\sqrt{\frac{h_{k,n}p_{k,n}}{\displaystyle\sum_{i\in\mathcal{I}_k}x_{i,n}h_{i,n}p_{i,n}+\sigma^2W/N}}\right)^2\right)},
\end{multline}
where $\mathcal{I}_k$ is the set of users that interfere on the $k$th user, and based on the optimal successive decoding, it is given by $\mathcal{I}_k=\{i:w_i>w_k\}$. In other words, each user will experience interference only from the users with higher weights. The suboptimal subcarrier and power allocation algorithm consists of two phases. In the first phase, the utilities of every user on each set of subcarriers are calculated. To this end, the MRT-WF single-user power allocation algorithm and partitioning scheme developed in the previous section are used. In the second phase, a set of subcarriers is allocated to one user based on the utilities calculated in the first phase. The algorithm iteratively allocates a set of subcarriers in each iteration. Analogous to the optimal algorithm, the allocated subcarriers are no longer available to the users from the same group.

Starting with zero interference, each user will calculate the utility on every set of subcarriers. Based on the calculated utilities, one user will be allocated one set of subcarriers. The interference is updated according to the subcarrier allocation and the users' utilities are recalculated. The algorithm iterate until all the users have zero-utility on all the unallocated subcarriers.

Two different subcarrier allocation criteria are considered for the second phase of the algorithm. The first one is allocating the set of subcarriers to the user that has the maximum utility value, $u_{k,m}$, and we refer to this subcarrier allocation criterion as SA$1$. This criterion is inspired by the fact that in the optimal algorithm, maximizing the main objective is decoupled into maximizing over every channel state. The second criterion, which is referred to as SA$2$, is to allocate a set of subcarriers to the user that achieve the maximum increase in the objective function~\eqref{Eq_LDS_Obj}
\begin{equation}
k^\star=\arg\max_k w_k (R_k-R_k^a),
\end{equation}
where $R_k^a$ is the rate of user $k$ using the subcarriers that already allocated to that user $(\mathcal{N}_k)$, and $R_k$ is the rate on all the subcarriers $(\mathcal{N}_k\cup \mathcal{N}_k^u)$ calculated at the power allocation phase.

The second algorithm, MUWF, has the same structure as the first algorithm but with different utility function which is defined at each subcarrier instead of set of subcarriers as follows
\begin{equation}\label{Eq_UtilityMUWF}
u_{k,n}=w_k \log{\left(1+\frac{h_{k,n}p_{k,n}}{\displaystyle\sum_{i\in\mathcal{I}_k}x_{i,n}h_{i,n}p_{i,n}+\sigma^2W/N}\right)}.
\end{equation}
Also, single-user water-filling is used for power allocation and one subcarrier is allocated in each iteration. After allocating all the subcarriers to the users, each user partitions the allocated subcarriers using the proposed partitioning scheme~\eqref{Eq_GreedyScPartition}, and the MRT-WF power allocation is implemented on the allocated subcarriers.
\begin{algorithm}[!t]
\small
\caption{\small Iterative Subcarrier and Power Allocation (MUMRT)}
\label{ALG_MUWF1}
\begin{algorithmic}[1]
\STATE \textbf{Users Grouping:} Set $w_k=\bar{w}_l,\quad\forall k\in\mathcal{B}_l,\quad l=1, \dots, d_c$.
\STATE \textbf{Initialization:} Put $\mathbf{J}=\textbf{0}$, $\mathcal{N}_k=\emptyset$ and $\mathcal{N}_k^u=\mathcal{N},\hspace{0.3cm}\forall k\in\mathcal{K}$.
\REPEAT
\STATE \textbf{Subcarriers Partitioning:} Partition the subcarriers using~\eqref{Eq_GreedyScPartition}:\\
$\mathcal{N}_k\rightarrow\mathcal{M}_k,\ \text{and}\ \mathcal{N}_k^u\rightarrow\mathcal{M}_k^u,\quad \forall k\in\mathcal{K}$.
\STATE \textbf{Power Allocation:} Considering the interference from other users, $\textbf{J}$, each user performs MRT-WF over $(\mathcal{N}_k\cup \mathcal{N}_k^u)$.
\STATE \textbf{Subcarrier Allocation (SA):} Calculate $u_{k,m}$ using~\eqref{Eq_UtilityMUMRT}$\ \forall k$:\\
\textbf{Option 1 (SA1):} Choose the pair such that:\\ $(k^\star,m^\star)=\arg\displaystyle\max_{k\in\mathcal{K},m\in\mathcal{M}_k^u} u_{k,m}$.\\
\textbf{Option 2 (SA2):} Perform power allocation for all $k\in\mathcal{K},\ n\in\mathcal{N}_k$ and calculate $R^a_k$. Choose the pair $(k^\star,m^\star)$ such that:\\
$k^\star=\arg\displaystyle\max_k w_k (R_k-R_k^a)$ and $m^\star=\arg\displaystyle\max_{m\in\mathcal{M}_{k^\star}^u} u_{k^\star,m}.$
\STATE Allocate the subcarriers subset $\mathcal{D}_{k^\star, m^\star}$ to user $k^\star$:
\\Set $x_{k^\star,n}=1,\quad\mathcal{N}_{k^\star}=\mathcal{N}_{k^\star}\cup \{n\},\quad \forall n\in\mathcal{D}_{k^\star, m^\star}$.
\STATE Remove the allocated subcarriers $(\mathcal{D}_{k^\star, m^\star})$ from the available subcarriers of the users' group $(\mathcal{B}_l)$ that includes $k^\star$:\\
$\mathcal{N}_k^u=\mathcal{N}_k^u\setminus{n},\quad\forall n\in\mathcal{D}_{k^\star, m^\star},\quad\forall k\in\mathcal{B}_l$.
\STATE Update the interference matrix elements $j_{k,n}\;\forall k\in\mathcal{I}_{k^\star}^c$,\\ $n\in\mathcal{D}_{k^\star, m^\star}$, where $\mathcal{I}_{k}^c=\mathcal{K}\setminus(\mathcal{I}_{k}\cup \{k\})$.
\UNTIL{$u_{k,m}=0,\quad\forall m\in\mathcal{M}_k^u$} or $\mathcal{M}_k^u=\emptyset,$ $\quad\forall k\in\mathcal{K}$.
\end{algorithmic}
\end{algorithm}
\begin{algorithm}[!t]
\small
\caption{\small Iterative Subcarrier and Power Allocation (MUWF)}
\label{ALG_MUWF2}
\begin{algorithmic}[1]
\STATE \textbf{Users Grouping:} Set $w_k= \bar{w}_l,\quad \forall k\in\mathcal{B}_l,\quad l=1, \dots, d_c$.
\STATE \textbf{Initialization:} Put $\mathbf{J}=\textbf{0}$, $\mathcal{N}_k=\emptyset$ and $\mathcal{N}_k^u=\mathcal{N},\hspace{0.3cm}\forall k\in\mathcal{K}$.
\REPEAT
\STATE \textbf{Power Allocation:} Considering the interference from other users, $\textbf{J}$, each user performs water-filling over $\mathcal{N}_k\cup \mathcal{N}_k^u$.
\STATE \textbf{Subcarrier Allocation (SA):} Calculate $u_{k,n}$ using~\eqref{Eq_UtilityMUWF}$\ \forall k$:\\
\textbf{Option 1 (SA1):} Choose the pair $(k^\star,n^\star)$ such that:\\ $(k^\star,n^\star)=\arg\displaystyle\max_{k\in\mathcal{K},m\in\mathcal{N}_k^u} u_{k,n}$.\\
\textbf{Option 2 (SA2):} Perform power allocation for all $k\in\mathcal{K},\ n\in\mathcal{N}_k$ and calculate $R^a_k$. Choose the pair $(k^\star,n^\star)$ such that:\\
$k^\star=\arg\displaystyle\max_k w_k (R_k-R_k^a)$ and $n^\star=\arg\displaystyle\max_{m\in\mathcal{N}_{k^\star}^u} u_{k^\star,n}.$
\STATE Allocate subcarrier $n^\star$ to user $k^\star$:
\\Set $x_{k^\star,n^\star}=1, \quad \text{and} \quad\mathcal{N}_{k^\star}=\mathcal{N}_{k^\star}\cup \{n^\star\}$.
\STATE Remove the subcarrier $n^\star$ from the available subcarriers of the users' group $(\mathcal{B}_l)$ that includes $k^\star$:
$\mathcal{N}_k^u\hspace{-0.05cm}=\mathcal{N}_k^u\hspace{-0.02cm}\setminus\{n^\star\},\ \forall k\in\mathcal{B}_l$.
\STATE Update the interference matrix elements $j_{k,n^\star}\;\forall k\in\mathcal{I}_{k^\star}^c$,
where $\mathcal{I}_{k}^c=\mathcal{K}\setminus(\mathcal{I}_{k}\cup \{k\})$.
\UNTIL{$u_{k,n}=0,\quad\forall n\in\mathcal{N}_k^u$} or $\mathcal{N}_k^u=\emptyset,$ $\quad\forall k\in\mathcal{K}$.
\STATE \textbf{Final power allocation}: Partition the subcarriers using~\eqref{Eq_GreedyScPartition} and implement MRT-WF for each user on the subcarriers allocated to it.
\end{algorithmic}
\end{algorithm}

The proposed subcarrier and power allocation algorithms, MUMRT and MUWF, are summarized in Algorithm~\ref{ALG_MUWF1}~and~\ref{ALG_MUWF2}, respectively. In Algorithm~\ref{ALG_MUWF1}~and~\ref{ALG_MUWF2}, the notation $\mathcal{N}_k^u$ represents the set of subcarriers that are not allocated to user $k$ or any user in its group. $\mathcal{M}_k^u$ is the set of symbols after subcarriers partitioning of $\mathcal{N}_k^u$. The interference matrix, $\mathbf{J}\in\mathcal{R}_{+}^{K\times N}$, indicates the interference every user will experience on each subcarrier from other users, where $j_{k,n}=\sum_{i\in\mathcal{I}_k}x_{i,n}h_{i,n}p_{i,n}$. It is worth mentioning that the algorithms have a deterministic number of iterations to allocate all the subcarriers for a given $d_c$ and $d_v$. The first algorithm, as it allocates $d_v$ subcarriers in each iteration, requires $N d_c/d_v$ maximum number of iterations to allocate all the available subcarriers. On the other hand, by allocating one subcarrier in each iteration, the second algorithm requires $N d_c$ maximum number of iterations.
\section{Numerical Results}\label{SectionFive}
\begin{figure}[!t]
\centering
\includegraphics[width=8.5cm]{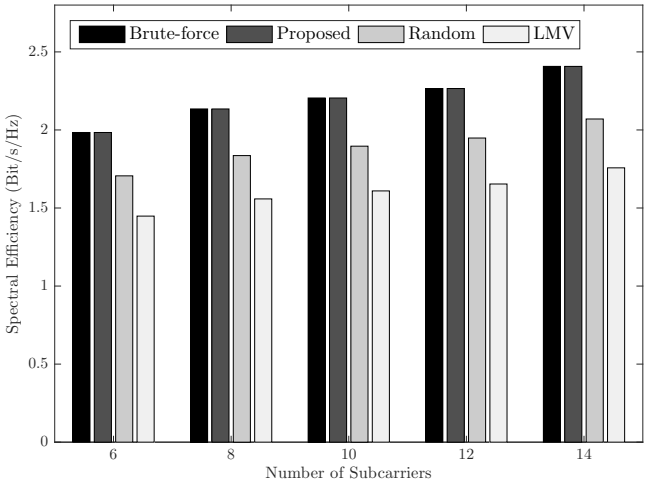}
\caption{Spectral efficiency comparison for different subcarrier partitioning schemes (with MRT-WF and $d_v=2$).}
\label{fig_PartitionMRTRes_dv2}
\end{figure}
\begin{figure}[!t]
\centering
\includegraphics[width=8.5cm]{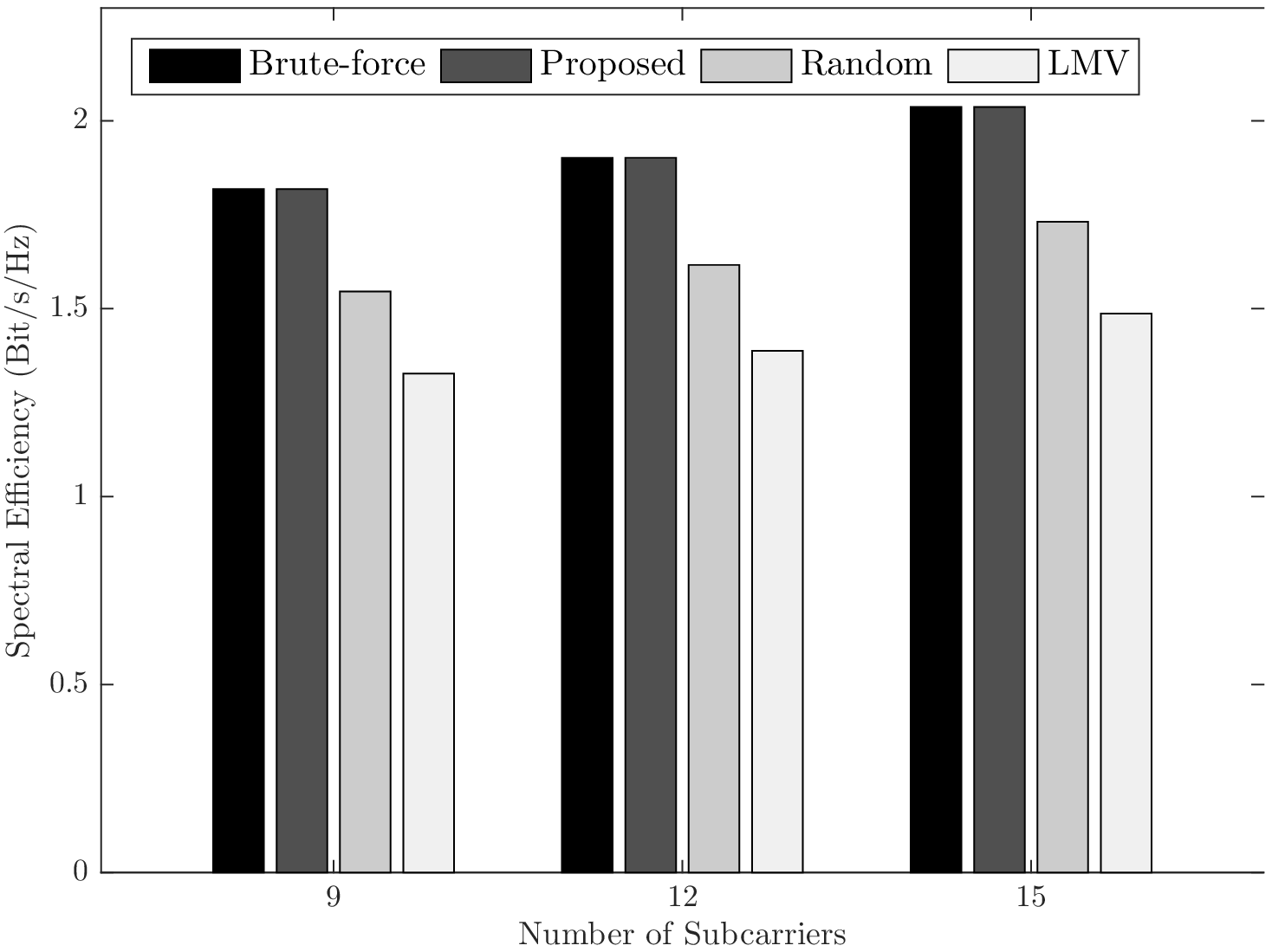}
\caption{Spectral efficiency comparison for different subcarrier partitioning schemes (with MRT-WF and $d_v=3$).}
\label{fig_PartitionMRTRes_dv3}
\end{figure}
\begin{figure}[!t]
\centering
\includegraphics[width=8.5cm]{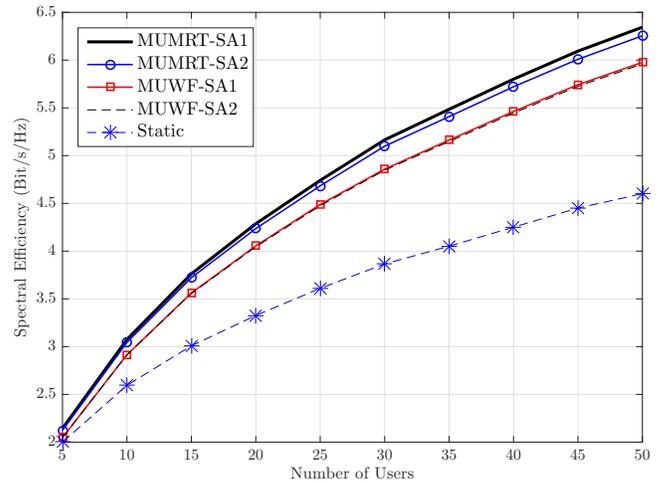}
\caption{Spectral efficiency comparison for the two algorithms with different subcarrier allocation criteria ($d_c=6$ and $d_v=2$).}
\label{fig_SE_MRTvsWF_SA1vsSA2}
\end{figure}
\begin{figure}[!t]
\centering
\includegraphics[width=8.5cm]{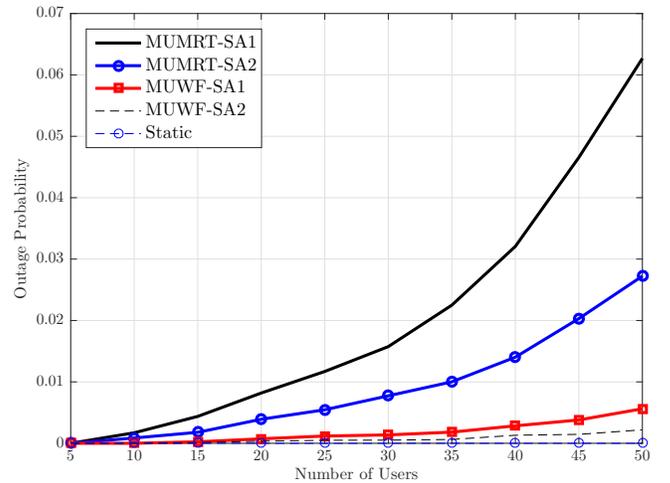}
\caption{Outage probability comparison for the two algorithms with different subcarrier allocation criteria ($d_c=6$ and $d_v=2$).}
\label{fig_OutageZero_MRTvsWF_SA1vsSA2}
\end{figure}

In this section, the performance of the proposed algorithms is evaluated through Monte Carlo simulation. A single hexagonal cell with $1$~km inradius is considered, and the users' locations are randomly generated and uniformly distributed over the cell. The maximum transmit power of each user is $23$~dBm and the system bandwidth is $5$~MHz consisting of $30$~subcarriers. The link gain between the base-station and a user is given as the product of path loss and fast fading effects. ITU pedestrian~B channel model~\cite{IMT_Guideline} is adopted for fast fading generation and the simplified model~\cite{WC_GS_Book} for the path loss. The noise power spectral density is assumed to be $-173$~dBm/Hz. Spectral efficiency and outage probability are used as the performance evaluation metrics. Taking into account fairness among users, the outage probability is defined as the number of users with zero rates divided by total number of users. To impose fairness in the system, users are grouped into $d_c$ groups based on their location from the base-station, and the group of users that are furthest from the base-station are assigned higher weight comparing to users that are closer to the base-station.

Starting with the single-user case, Figs.~\ref{fig_PartitionMRTRes_dv2} and~\ref{fig_PartitionMRTRes_dv3} show the user spectral efficiency for different subcarriers partitioning schemes with $d_v=2$ and~$3$, respectively. Due to the complexity of brute-force search, small numbers of subcarriers ($N$) have been considered in this case. As it can be observed from the figures, the proposed partitioning scheme achieves the same spectral efficiency as the brute-force search and significantly outperforms the random and LMV schemes. This provide a numerical proof on the optimality of the proposed scheme. The random and LMV schemes only achieve about $86\%$ and $73\%$ of the optimal solution, respectively, which indicates the sub-optimality of these schemes. Furthermore, the very poor performance of the LMV partitioning, demonstrate the importance of vector majorization in generating the symbols gains, where it shows that a less majorized gain vector will produce very low user rate.

\begin{figure}[!t]
\centering
\includegraphics[width=8.5cm]{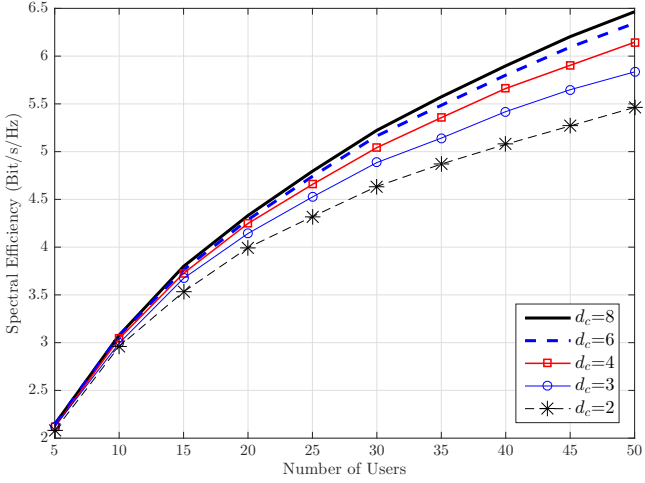}
\caption{Spectral efficiency comparison for different subcarrier loading $(d_c)$ (with MUMRT-SA1 and $d_v=2$).}
\label{fig_SE_MUMRT_SA1_diffDc}
\end{figure}
\begin{figure}[!t]
\centering
\includegraphics[width=8.5cm]{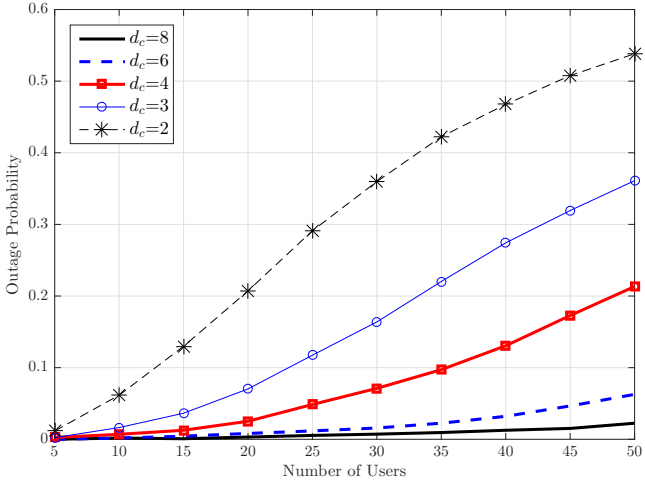}
\caption{Outage probability comparison for different subcarrier loading $(d_c)$ (with MUMRT-SA1 and $d_v=2$).}
\label{fig_OutageZero_MUMRT_SA1_diffDc}
\end{figure}
\begin{figure}[!t]
\centering
\includegraphics[width=8.5cm]{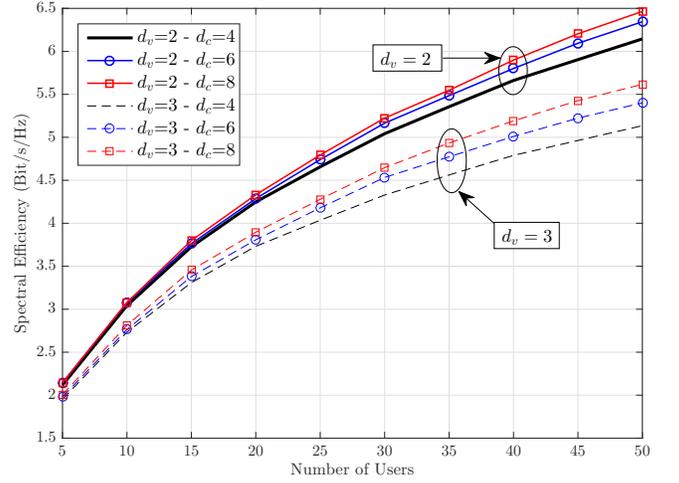}
\caption{Spectral efficiency comparison for different effective spreading factors (with MUMRT-SA$1$).}
\label{fig_Spectral_Eff_2}
\end{figure}
\begin{figure}[!t]
\centering
\includegraphics[width=8.5cm]{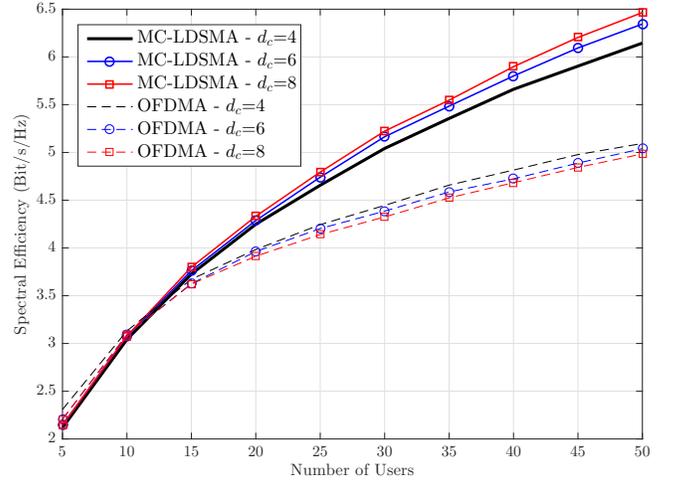}
\caption{Spectral efficiency comparison between MC-LDSMA (with MUMRT-SA$1$ and $d_v=2$) and OFDMA systems.}
\label{fig_Spectral_Eff_3}
\end{figure}

Figs.~\ref{fig_SE_MRTvsWF_SA1vsSA2}~and~\ref{fig_OutageZero_MRTvsWF_SA1vsSA2} present the spectral efficiency and outage probability, respectively, in the multiuser case, for the two proposed algorithms together with the static subcarrier and power allocation case. The figures show the difference between the two algorithms (MUMRT and MUWF) and the two subcarrier allocation criteria (SA$1$ and SA$2$). In the static case, the subcarriers are equally distributed among the users, and the power of each user is equally allocated on its subcarriers. It is evident from Fig.~\ref{fig_SE_MRTvsWF_SA1vsSA2} that both algorithms significantly improve the system spectral efficiency comparing to the static subcarrier and power allocation. Also, as the figure shows, the MUMRT algorithm achieves higher spectral efficiency compared with the MUWF algorithm. As shown in Fig.~\ref{fig_OutageZero_MRTvsWF_SA1vsSA2}, MUWF algorithm achieves less outage probability thanks to the capability of allocating individual subcarriers to the users. Furthermore, it can be observed from Fig.~\ref{fig_SE_MRTvsWF_SA1vsSA2} that the two subcarrier allocation criteria achieve almost the same spectral efficiency. The SA$1$ criterion has a marginal gain compared with SA$2$ in MUMRT algorithm. However, as it can be seen from Fig.~\ref{fig_OutageZero_MRTvsWF_SA1vsSA2}, SA$2$ provides better fairness among users compared to SA$1$ by achieving less outage probability. Consequently, it can be concluded that SA$1$ is more greedy in the resource allocation in comparison to SA$2$.

To demonstrate the effect of subcarrier loading, $d_c$, Figs.~\ref{fig_SE_MUMRT_SA1_diffDc}~and~\ref{fig_OutageZero_MUMRT_SA1_diffDc} show the spectral efficiency and outage probability, respectively, for different subcarrier loading values. It can be seen from the figures that increasing the subcarrier loading improves the system performance in terms of spectral efficiency and outage probability at the cost of increased receiver complexity. Increasing the subcarrier loading will allow subcarriers to be reused by more users, which will increase the spectral efficiency. Also, the number of users' groups will be increased, resulting in improved fairness to the system, which consequently improves the outage probability. Furthermore, it can be observed that the required subcarrier loading increases as the total number of users increase to maintain a desired system performance. As an example, when the total number of users is $50$ and the subcarrier loading are $2$ and $3$, the outage probabilities will be $54\%$ and $36\%$, respectively. Therefore, the subcarrier loading needs to be adapted to the system load (total number of users) to balance between the system performance and complexity.

In order to quantify the effect of the effective spreading factor $d_v$ on the system performance, Fig.~\ref{fig_Spectral_Eff_2} presents the spectral efficiency for $d_v=2$ and $d_v=3$. As is evident from the figure, increasing the effective spreading factor will decrease the spectral efficiency. Also, the subcarrier loading has more impact for higher spreading factors. This result may give indication that it is more spectral efficient to not spread and use an orthogonal scheme. The following discussion will clarify this issue by comparing between MC-LDSMA and orthogonal transmission.

\begin{figure}[!t]
\centering
\includegraphics[width=8.5cm]{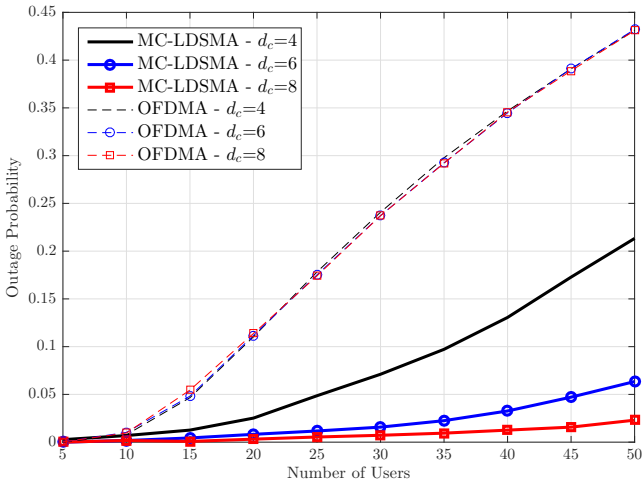}
\caption{Outage probability comparison between MC-LDSMA (with MUMRT-SA$1$ and $d_v=2$) and OFDMA systems.}
\label{fig_Outage_Pro_3}
\end{figure}
From system level perspective, spreading has advantage and disadvantage. The advantage of spreading is to allow the users share the subcarriers at the same time, which will increase the system capacity comparing to orthogonal transmission. On the other hand, spreading reduces the user's rate~\cite{Massey_ITspread}. Thus, it is essential to investigate if the gain of allowing the users to share the subcarriers is sufficient to compensate the loss due to spreading. To this end, as a non-orthogonal transmission, MC-LDSMA will be compared with OFDMA system, as an orthogonal transmission technique. For OFDMA system, the subcarrier and power allocation algorithm developed in~\cite{EURASIP_OFDMA} will be used.

Figs.~\ref{fig_Spectral_Eff_3}~and~\ref{fig_Outage_Pro_3} show the spectral efficiency and outage probability, respectively, for MC-LDSMA and OFDMA with different subcarrier loading $(d_c)$. For the OFDMA system, the $(d_c)$ value is just for the weight calculation to keep the comparison fair and it does not affect the subcarrier loading as it is always equal to $1$. As is clear from Fig.~\ref{fig_Spectral_Eff_3}, for very small number of users ($5$ and $10$), OFDMA slightly outperforms MC-LDSMA. However, for high number of users, MC-LDSMA significantly outperforms OFDMA. This due the fact that when the total number of users is small there is no need for reusing the subcarriers, and the available subcarriers are enough to support all the users.

On the other hand, when the total number of users is high, the subcarrier sharing is essential to support more users and achieve higher spectral efficiency. Furthermore, as can be observed from the figure, using higher fairness by increasing the number of users groups, results in a marginal reduction in the OFDMA spectral efficiency. This is because more fairness will be imposed and the subcarrier loading is not changed. In contrast, in MC-LDSMA, the increase in fairness is combined with an increase in the subcarrier loading, which improves the system performance. Also, as Fig.~\ref{fig_Outage_Pro_3} shows, MC-LDSMA is fairer comparing to OFDMA system, and the resulting outage probability is much less than OFDMA. As an example, when the total number of users is $50$, OFDMA has outage probability about $43\%$, while for MC-LDSMA $21\%,\ 6.3\% \text{ and } 2.3\%$ can be achieved with $d_c=4, 6 \text{ and } 8$, respectively.

From the results shown in Figs.~\ref{fig_Spectral_Eff_3} and~\ref{fig_Outage_Pro_3}, it can be concluded that low density spreading with the proposed radio resource allocation algorithms can significantly increase the system performance in terms of spectral efficiency and outage probability compared with orthogonal transmission techniques. These gains can be translated into an increase in the number of supported users within the cell and higher users' data rates.
\section{Conclusion}\label{SectionSix}
In this paper, we considered the problem of radio resource allocation for single-cell MC-LDSMA system in single-user and multiuser scenarios. For single-user radio resource allocation, an optimal power allocation has been derived, which consists of two steps: water-filling over the symbols and maximum ratio transmission over the subcarriers of each symbol. Also, we have shown that the optimal subcarriers partitioning scheme is the one that gives the most majorized symbols' gain vector. Numerical analysis revealed that the proposed partitioning scheme can significantly improve the user rate comparing to the conventional random partitioning. For multiuser radio resource allocation in MC-LDSMA, an analysis for the optimality conditions is provided and an optimal power allocation for a relaxed problem has been derived for continuous frequency selective channel. Using the structure of the optimal solution and the single-user power allocation, two suboptimal subcarrier and power allocation algorithms are proposed. The main difference between the two algorithms is the spreading effect: the first algorithm (MUMRT) takes into account the spreading from the onset of the allocation process and the subcarriers are allocated on subsets basis, whereas the second algorithm (MUWF) assumes no spreading and the subcarriers are allocated individually. The performance of the proposed algorithms have been evaluated by simulations, and results show that the proposed algorithms significantly improve the system spectral efficiency comparing to static subcarrier and power allocation. MUMRT algorithm has shown better performance comparing to the other one, and one subcarrier allocation (SA$1$) criterion has shown more greedy nature comparing to the other one (SA$2$). It is shown that there is a trade-off between the complexity (subcarrier loading) and the system performance, and the subcarrier loading can be chosen accordingly based on the system loading. Furthermore, we found that increasing the spreading factor to more than $2$ will decrease the system performance. Finally, a comparison between MC-LDSMA and OFDMA systems has been carried out and shown that MC-LDSMA can improve the system performance in terms of spectral efficiency and outage probability comparing to OFDMA.

\appendices
\section{Proof of Theorem~\ref{ThrmRateSchur}}\label{App_RateScurProof}
\begin{IEEEproof}
The proof of the theorem makes use of the properties of the elementary symmetric polynomials by using them as bridge to link between symbols gains majorization and the user rate. As the relation between the elementary symmetric polynomials and vector majorization is known, the only thing required is to show the connection between the user rate and the elementary symmetric polynomials. To this end, the user rate will be reformulated as a function of the elementary symmetric polynomials. The water-level in~\eqref{Eq_WaterLevel_1} can be rewritten as
\begin{equation}
\frac{1}{\lambda_k}=\frac{1}{M_k}\left(P_k+\frac{\displaystyle\sum_{m=1}^{M_k}\prod_{i=1\backslash m}^{M_k} g_{k,i}}{\displaystyle\prod_{m=1}^{M_k} g_{k,m}}\right),
\end{equation}
which can be further simplified using elementary symmetric polynomials as follows
\begin{equation}
\frac{1}{\lambda_k}=\frac{S_{M_k}(\mathbf{g}_k)P_k+S_{M_k-1}(\mathbf{g}_k)}{M_k S_{M_k}(\mathbf{g}_k)},
\end{equation}
where $S_i(.)$ is the $i$th elementary symmetric polynomial and given by~\cite{SymmetricFunBook}
\begin{equation}
S_{M_k}(\mathbf{g}_k)=\prod_{m=1}^{M_k} g_{k,m},
\end{equation}
\begin{equation}
S_{M_k-1}(\mathbf{g}_k)=\sum_{m=1}^{M_k}\prod_{i=1\backslash m}^{M_k} g_{k,i}.
\end{equation}
Consequently, the user rate~\eqref{Eq_SymbolsRate_2} will be
\begin{eqnarray}\label{Eq_SymbolsRate_3}
R_k(\mathbf{g}_k)&=& \sum_{m\in\mathcal{M}_k}\log{\left(\frac{\bar{g}_{k,m}}{\lambda_k}\right)}=\log{\left(\prod_{m=1}^{M_k}\frac{\bar{g}_{k,m}}{\lambda_k}\right)}\nonumber\\
&=& \log{\left(\frac{1}{\lambda_k^{M_k}}S_{M_k}(\mathbf{g}_k)\right)}\nonumber\\
&=& M_k\log{\left(\frac{1}{\lambda_k}\right)}+\log{\left(S_{M_k}(\mathbf{g}_k)\right)}\nonumber\\
&=& M_k\log{\left(\frac{S_{M_k}(\mathbf{g}_k)P_k+S_{M_k-1}(\mathbf{g}_k)}{M_k S_{M_k}(\mathbf{g}_k)}\right)}\nonumber\\
&&+\log{\left(S_{M_k}(\mathbf{g}_k)\right)}\nonumber\\
&=& M_k\log{\Big(S_{M_k}(\mathbf{g}_k)P_k+S_{M_k-1}(\mathbf{g}_k)\Big)}\nonumber\\
&&-(M_k-1)\log{\left(S_{M_k}(\mathbf{g}_k)\right)}-M_k\log{(M_k)}.\nonumber\\
\end{eqnarray}
Both $S_{M_k}(\mathbf{g}_k)$ and $S_{M_k-1}(\mathbf{g}_k)$ are Schur-concave functions~\cite[Chapter~3]{Major_Marshall_Book}, and will be decreased by majorizing $\mathbf{g}_k$.
Now, the change in the user rate with respect to $S_{M_k}(\mathbf{g}_k)$ and $S_{M_k-1}(\mathbf{g}_k)$ can be found through the partial derivative of the user rate as follows
\begin{equation}\label{Eq_RateDerivative1}
\frac{\partial R_k(\mathbf{g}_k)}{\partial S_{M_k}}=\frac{S_{M_k}(\mathbf{g}_k)P_k-(M_k-1)S_{M_k-1}(\mathbf{g}_k)}{S^2_{M_k}(\mathbf{g}_k)P_k+S_{M_k}(\mathbf{g}_k)S_{M_k-1}(\mathbf{g}_k)},
\end{equation}
\begin{equation}\label{Eq_RateDerivative2}
\frac{\partial R_k(\mathbf{g}_k)}{\partial S_{M_k-1}}=\frac{M_k}{S^2_{M_k}(\mathbf{g}_k)P_k+S_{M_k}(\mathbf{g}_k)S_{M_k-1}(\mathbf{g}_k)}.
\end{equation}
Consequently, the user rate will be a decreasing function of majorizing $\mathbf{g}_k$ if the total change with respect to $S_{M_k}(\mathbf{g}_k)$ and $S_{M_k-1}(\mathbf{g}_k)$ is negative, i.e.
\begin{equation}
\frac{\partial R_k(\mathbf{g}_k)}{\partial S_{M_k}}+\frac{\partial R_k(\mathbf{g}_k)}{\partial S_{M_k-1}}<0,
\end{equation} 
which is satisfied if the following condition holds
\begin{equation}
P_k<(M_k-1)\frac{S_{M_k-1}(\mathbf{g}_k)}{S_{M_k}(\mathbf{g}_k)}-M_k.
\end{equation}
In fact, $\frac{S_{M_k-1}(\mathbf{g}_k)}{S_{M_k}(\mathbf{g}_k)}$ is a Schur-convex function of $\mathbf{g}_k$~\cite[Chapter~3]{Major_Marshall_Book}, hence, the minimum its achieved when all the elements in $\mathbf{g}_k$ are equal, and the upper bound on the user power will be
\begin{equation}\label{Eq_PowerBound}
P_k<(M_k-1)\frac{\frac{M_k}{M_k^{M_k-1}}}{\frac{1}{M_k^{M_k}}}-M_k\quad \Longrightarrow\quad P_k<M_k^3-M_k^2-M_k.
\end{equation}
This implies that if the user total power $(P_k)$ is less than the bound in \eqref{Eq_PowerBound}, the user rate will be decreased with the increase in $S_{M_k}(\mathbf{g}_k)$ and $S_{M_k-1}(\mathbf{g}_k)$. Consequently, the user rate~\eqref{Eq_SymbolsRate_2} will be increased by majorizing $\mathbf{g}_k$.
\end{IEEEproof}
\section{Proof of Lemma~\ref{Lem_SC_Partitiong}}\label{App_SymMMVproof}
\begin{IEEEproof}
The optimality proof of this scheme can be done by showing that the proposed partitioning scheme will give symbols gains vector that majorize any other partitioning scheme. Vectors majorization is defined as follows.
\begin{defn}\label{Def_Majorization}
\textbf{(Majorization of vectors~\cite[Chapter~1]{Major_Marshall_Book}):} For two vectors $\mathbf{x}, \mathbf{y}\in\mathcal{R}^{N_k}$ with descending order such that $x_1\geq x_2\geq\dots \geq x_{N_k}\geq 0$ and $y_1\geq y_2\geq\dots \geq y_{N_k}\geq 0$, $\mathbf{x}$ is said to be majorize $\mathbf{y}$ and written as $\mathbf{x}\succ\mathbf{y}$ if
\begin{displaymath}
\sum_{i=1}^n x_i \geq \sum_{i=1}^n y_i,\quad n=1,\dots, N_k-1\quad \text{and}\quad \sum_{i=1}^{N_k} x_i = \sum_{i=1}^{N_k} y_i.
\end{displaymath}
\end{defn}
Let $\pi$ represents the permutation of the subcarriers gains in descending order and $\rho$ is any other permutation. Based on these permutations it is clear that
\begin{equation}
\sum_{i=1}^m x_i \geq \sum_{i=1}^m y_i,\quad m=1,\dots, M_k-1\quad \text{and}\quad \sum_{i=1}^{M_k} x_i= \sum_{i=1}^{M_k} y_i,
\end{equation}
where $x_i=\{g_{k,\pi(i)}, g_{k,\pi(i+1)}, \dots,\ g_{k,\pi(i+d_v-1)}\}$ and $y_i=\{g_{k,\rho(i)}, g_{k,\rho(i+1)}, \dots,\ g_{k,\rho(i+d_v-1)}\}$. Consequently, vector $\mathbf{x}$ is majorizing vector $\mathbf{y}$ according to Definition~\ref{Def_Majorization}.
\end{IEEEproof}

\ifCLASSOPTIONcaptionsoff
\newpage
\fi

\bibliographystyle{IEEEtran}
\bibliography{IEEEabrv,References}

\begin{thebibliography}{10}
\providecommand{\url}[1]{#1}
\csname url@samestyle\endcsname
\providecommand{\newblock}{\relax}
\providecommand{\bibinfo}[2]{#2}
\providecommand{\BIBentrySTDinterwordspacing}{\spaceskip=0pt\relax}
\providecommand{\BIBentryALTinterwordstretchfactor}{4}
\providecommand{\BIBentryALTinterwordspacing}{\spaceskip=\fontdimen2\font plus
\BIBentryALTinterwordstretchfactor\fontdimen3\font minus
  \fontdimen4\font\relax}
\providecommand{\BIBforeignlanguage}[2]{{%
\expandafter\ifx\csname l@#1\endcsname\relax
\typeout{** WARNING: IEEEtran.bst: No hyphenation pattern has been}%
\typeout{** loaded for the language `#1'. Using the pattern for}%
\typeout{** the default language instead.}%
\else
\language=\csname l@#1\endcsname
\fi
#2}}
\providecommand{\BIBdecl}{\relax}
\BIBdecl

\bibitem{Tse_LDS}
A.~Montanari and D.~Tse, ``Analysis of belief propagation for non-linear
  problems: The example of {CDMA} (or: How to prove {T}anaka's formula),'' in
  \emph{IEEE Information Theory Workshop}, March 2006, pp. 160--164.

\bibitem{YoshidaTanaka_SSCDMA}
M.~Yoshida and T.~Tanaka, ``Analysis of sparsely-spread {CDMA} via statistical
  mechanics,'' in \emph{IEEE International Symposium on Information Theory},
  July 2006, pp. 2378--2382.

\bibitem{RaymondSaad_SSCDMA}
J.~Raymond and D.~Saad, ``Sparsely spread {CDMA} - a statistical
  mechanics-based analysis,'' \emph{Journal of Physics A: Mathematical and
  Theoretical}, vol.~40, pp. 12\,315--12\,333, 2007.

\bibitem{GuoWang_SSCDMA}
D.~Guo and C.-C. Wang, ``Multiuser detection of sparsely spread {CDMA},''
  \emph{{IEEE} J. Sel. Areas Commun.}, vol.~26, no.~3, pp. 421--431, April
  2008.

\bibitem{FerThesis}
F.~Wathan, ``Low-density signature communications in overloaded {CDMA}
  systems,'' Ph.D. dissertation, University of Surrey, 2007.

\bibitem{Verdu_ISI}
R.~Cheng and S.~Verdu, ``Gaussian multiaccess channels with {ISI}: capacity
  region and multiuser water-filling,'' \emph{{IEEE} Trans. Inf. Theory},
  vol.~39, no.~3, pp. 773--785, May 1993.

\bibitem{Tse_Polymatroid}
D.~Tse and S.~Hanly, ``Multiaccess fading channels-{P}art {I}: Polymatroid
  structure, optimal resource allocation and throughput capacities,''
  \emph{{IEEE} Trans. Inf. Theory}, vol.~44, no.~7, pp. 2796--2815, Nov. 1998.

\bibitem{WeiYu_Jour}
W.~Yu, W.~Rhee, S.~Boyd, and J.~Cioffi, ``Iterative water-filling for
  {G}aussian vector multiple-access channels,'' \emph{{IEEE} Trans. Inf.
  Theory}, vol.~50, no.~1, pp. 145--152, Jan. 2004.

\bibitem{NgSung_SPA_UL}
C.~Y. Ng and C.~W. Sung, ``Low complexity subcarrier and power allocation for
  utility maximization in uplink {OFDMA} systems,'' \emph{{IEEE} Trans.
  Wireless Commun.}, vol.~7, no.~5, pp. 1667--1675, May 2008.

\bibitem{HuangBerry_SPA_UL}
J.~Huang, V.~Subramanian, R.~Agrawal, and R.~Berry, ``Joint scheduling and
  resource allocation in uplink {OFDM} systems for broadband wireless access
  networks,'' \emph{{IEEE} J. Sel. Areas Commun.}, vol.~27, no.~2, pp.
  226--234, Feb. 2009.

\bibitem{EURASIP_OFDMA}
M.~Al-Imari, P.~Xiao, M.~Imran, and R.~Tafazolli, ``Low complexity subcarrier
  and power allocation algorithm for uplink {OFDMA} systems,'' \emph{EURASIP
  Journal on Wireless Communications and Networking}, vol.~98, no.~1, pp. 1--6,
  April 2013.

\bibitem{YPW_Auction}
K.~Yang, N.~Prasad, and X.~Wang, ``An auction approach to resource allocation
  in uplink {OFDMA} systems,'' \emph{{IEEE} Trans. Signal Process.}, vol.~57,
  no.~11, pp. 4482--4496, Nov. 2009.

\bibitem{HJL_Nash_J}
Z.~Han, Z.~Ji, and K.~Liu, ``Fair multiuser channel allocation for {OFDMA}
  networks using {N}ash bargaining solutions and coalitions,'' \emph{{IEEE}
  Trans. Commun.}, vol.~53, no.~8, pp. 1366--1376, Aug. 2005.

\bibitem{Noh_Game}
W.~Noh, ``A distributed resource control for fairness in {OFDMA} systems:
  English-auction game with imperfect information,'' in \emph{IEEE Global
  Telecommunications Conference}, Dec. 2008, pp. 1--6.

\bibitem{MyPaper_VTC11}
M.~Al-Imari, M.~Imran, R.~Tafazolli, and D.~Chen, ``Subcarrier and power
  allocation for {LDS-OFDM} system,'' in \emph{IEEE Vehicular Technology
  Conference}, May 2011, pp. 1--5.

\bibitem{MyPaper_IWCMC}
------, ``Performance evaluation of low density spreading multiple access,'' in
  \emph{International Wireless Communications and Mobile Computing Conference},
  Aug. 2012, pp. 383--388.

\bibitem{WC_GS_Book}
A.~Goldsmith, \emph{Wireless Communications}.\hskip 1em plus 0.5em minus
  0.4em\relax Cambridge University Press, 2005.

\bibitem{IMT_Guideline}
ITU, \emph{Guidelines for Evaluation of Radio Transmission Technologies for
  {IMT}-2000}.\hskip 1em plus 0.5em minus 0.4em\relax Recommendation ITU-R
  M.1225, 1997.

\bibitem{Massey_ITspread}
J.~Massey, ``Information theory aspects of spread-spectrum communications,'' in
  \emph{IEEE Third International Symposium on Spread Spectrum Techniques and
  Applications}, vol.~1, July 1994, pp. 16--21.

\bibitem{SymmetricFunBook}
I.~G. Macdonald, \emph{Symmetric Functions and Hall Polynomials}.\hskip 1em
  plus 0.5em minus 0.4em\relax Oxford University Press, 2003.

\bibitem{Major_Marshall_Book}
A.~W. Marshall and I.~Olkin, \emph{Inequalities: Theory of Majorization and Its
  Applications}.\hskip 1em plus 0.5em minus 0.4em\relax New York: Academic
  Press, 1979.

\end{thebibliography}


\vspace{-1cm}
\begin{IEEEbiographynophoto}{Mohammed Al-Imari}
(S'10-M'14) received the B.Sc. degree in electronics and communication engineering (first rank) from the University of Technology, Baghdad, Iraq, in 2005; the M.Sc. degree (distinction) in wireless communication systems (first rank) from Brunel University, London, U.K., in 2008; and the Ph.D. degree in electronic engineering from the University of Surrey, Guildford, U.K., in 2013.

During 2013-2015, he was a Research Fellow with the Institute for Communication Systems/5G Innovation Centre, U.K. In October 2015, he joined Samsung Electronics R\&D Institute, Staines, U.K., as a 5G Researcher, where he actively contributes to the Horizon 2020 5G-PPP project FANTASTIC-5G. His research interests include information theory, multiple access techniques, and radio resource management in multicarrier communications.

Dr. Al-Imari received the IEEE Communications Prize for the Best Project from the Department of Electronic and Computer Engineering from Brunel University.
\end{IEEEbiographynophoto}

\vspace{-1cm}
\begin{IEEEbiographynophoto}{Muhammad Ali Imran}
(M'03-SM'12) received the M.Sc. (Distinction) and Ph.D. degrees from Imperial College London, U.K., in 2002 and 2007, respectively. He is currently a Reader (Associate Professor) in the Institute for Communication Systems (ICS - formerly known as CCSR) at the University of Surrey, UK and an Adjunct Associate Professor at University of Oklahoma, USA. He had a leading role in a number of multimillion international research projects encompassing the areas of energy efficiency, fundamental performance limits, sensor networks and self-organising cellular networks. He is also leading the new physical layer work area for 5G innovation centre and the curriculum development for the new Engineering for Health program at Surrey. He has a global collaborative research network spanning both academia and key industrial players in the field of wireless communications. He has supervised more than 20 successful PhD graduates. He has contributed to 10 patents and published over 200 peer-reviewed research papers including more than 20 IEEE Transactions. He has been awarded IEEE Comsoc's Fred Ellersick award 2014, FEPS Learning and Teaching award 2014 and twice nominated for Tony Jean's Inspirational Teaching award. He was a shortlisted finalist for The Wharton-QS Stars Awards 2014 for innovative teaching and VC's learning and teaching award in University of Surrey. He has acted as a guest editor for IET Signal Processing, IEEE Communications Magazine, IEEE Wireless Communication Magazine, IEEE Access. He is an associate Editor for IEEE Communications Letters, IEEE Access and IET Communications Journal. He is a senior member of IEEE and a Senior Fellow of Higher Education Academy (SFHEA) UK.
\end{IEEEbiographynophoto}

\begin{IEEEbiographynophoto}{Pei Xiao}
(SM'11) received the Ph.D. degree from Chalmers University of Technology, Sweden in 2004. Prior to joining the University of Surrey in 2011, he worked as a research fellow at Queen's University Belfast and had held positions at Nokia Networks in Finland. He is a Reader at University of Surrey and also the technical manager of 5G Innovation Centre (5GIC), leading and coordinating research activities in all the work areas in 5GIC (http://www.surrey.ac.uk/5gic/research). Dr Xiao's research interests and expertise span a wide range of areas in communications theory and signal processing for wireless communications.
\end{IEEEbiographynophoto}
\vfill
\end{document}